\documentclass[journal]{IEEEtran}
\usepackage{amsmath,amsthm,commath}
\usepackage{amssymb,textcomp}
\usepackage{cite}
\usepackage{tabularx}
\newtheorem{theorem}{Theorem}
\newtheorem{lemma}[theorem]{Lemma}
\usepackage[english]{babel}
\usepackage[utf8]{inputenc}
\usepackage{algorithm}
\usepackage[noend]{algpseudocode}
\usepackage{enumerate}
\usepackage{epsfig}
\usepackage{graphicx}
\usepackage{subcaption}
\usepackage{color}
\newtheorem{assumption}{Assumption}
\ifCLASSINFOpdf
\else
\fi
\usepackage{array}
\usepackage{fixltx2e}
\begin{document}
%
\title{Particle Filter for Randomly Delayed Measurements with Unknown Latency Probability}
%
%
%
\author{Ranjeet~Kumar~Tiwari, Shovan~Bhaumik, and Paresh~Date}
\maketitle

\begin{abstract}
This paper focuses on designing a particle filter for randomly delayed measurements with an unknown latency probability. A generalized measurement model is adopted which includes measurements that are delayed randomly by an arbitrary but fixed maximum number of the steps, along with random packet drops. Recursion equation for importance weights is derived under the presence of random delays. Offline and online algorithms for identification of the unknown latency parameter using the maximum likelihood criterion are proposed. Further, this work explores the conditions which ensure the convergence of the proposed particle filter. Finally, two numerical examples concerning problems of non-stationary growth model and the bearing-only tracking are simulated to show the effectiveness and superiority of the proposed filter.
\end{abstract}


%

\section{Introduction}
%
%
%
%
State estimation for non-linear discrete-time stochastic systems has been recently getting considerable attention from researchers\cite{charalampidis2011computationally,vsimandl2009derivative,julier2000new} because of its wide range of applications in various fields of science, including engineering \cite{lin2002comparison}, econometrics \cite{lopes2011particle} and meteorology \cite{stordal2011bridging}, for example. Bayesian approach \cite{ho1964bayesian} gives a recursive relation for the computation of the posterior probability density functions (pdf) of the unobserved states. But computation of the posterior pdf, in case of a non-linear system, are often numerically intractable and hence approximations of these pdf are often employed. Particle filters (PFs) are a set of powerful sequential Monte Carlo methods which can flexibly be designed under Bayesian framework to solve non-linear and non-Gaussian problems by approximating the posterior pdf empirically \cite{carpenter1999improved}. According to \cite{arulampalam2002tutorial}, particle filter outperforms its contemporary approximate Bayesian filters like the extended Kalman filter (EKF) and the grid-based filters in solving non-linear state estimation problems. However, most research on the EKF \cite{nelson1976simultaneous} and as well as on the traditional PF \cite{crisan2002survey,arulampalam2002tutorial,zuo2013adaptive} typically assumes that measurements are available at each time step without any delay. In practice, many fields including aerospace and underwater target tracking\cite{moose1985adaptive}, control applications\cite{kolmanovsky2001optimal}, communication\cite{nilsson1998stochastic} \textit{etc.} see random delays in receiving the measurements. This delay can be caused by the limitations of common network channel and it needs to be accounted while designing a filter.

In the literature, there exist a good number of research which has considered random delays for linear estimators. Ma \textit{et al.} \cite{ma2011optimal} proposed a linear estimator which deals with uncertain measurements and multiple packet dropouts along with random sensor delays. \cite{ma2016linear} has designed a linear estimator for networked system with one-step randomly delayed observations and multiple packet dropouts. A linear networked estimator \cite{yan2017networked} has been proposed recently to tackle irregularly-spaced and delayed measurements in the multi-sensor environment. On the other hand, research on random delays and packet drops for solving non-linear problems of estimation is limited and still developing. Hermoso-Carazo \textit{et al.} proposed an improved versions of EKF and UKF (unscented Kalman filter) for one-time step\cite{hermoso2007extended} and two-time step\cite{hermoso2009unscented} randomly delayed measurements. In \cite{singh2016quadrature}, quadrature filters have been modified to solve non-linear filtering problem with one-step randomly delayed measurements. Wang \textit{et al.}\cite{wang2013gaussian} used cubature Kalman filter (CKF)\cite{arasaratnam2009cubature} to tackle one-step randomly delayed measurements for non-linear systems. Later, Singh \textit{et al.} \cite{singh2017modified} proposed a methodology to solve non-linear estimation problems with multi-step randomly delayed measurements. However, all these non-linear filters are restricted to Gaussian approximations. Moreover, they have considered that latency probability of delayed measurements is known. Recently, \cite{zhang2016particle} discussed a modified PF that deals with one-step randomly delayed measurement with unknown latency probability and almost concurrently, the same authors presented a short work \cite{huang2015particle} in which they discussed a PF for multi-step randomly delayed measurements with known latency probability.  But, in these two works, there is no consideration of packet drops.

In this paper, we consider randomly delayed measurements along with a possibility of packet drops. Moreover, latency probability of the measurements which are supposed to be randomly delayed by up to a user-allocated maximum number of steps, is considered to be unknown. We develop the importance weight recursion that accounts for such randomness of delayed measurements. A method using the maximum-likelihood (ML) criterion is presented which identifies the unknown latency probability of delayed measurements and packet drops. Further, this work explores the conditions that ensure the convergence of the modified PF designed for randomly delayed measurements and packet drops.

With the help of two numerical examples, the effectiveness and superiority of the modified PF designed for arbitrary step randomly delayed measurements are demonstrated in comparison with the PF designed for non-delayed and one-step randomly delayed measurements.

 The rest of the paper is organized as follows. The problem statement is defined in section \textrm{II}. In section \textrm{III}, designing of modified PF is illustrated and later convergence is discussed. Section \textrm{IV} deals with identification of unknown latency probability. Sequence of steps in the form of algorithm has been presented for a clear picture of identification process. In Section \textrm{V}, simulation results are presented to demonstrate the superiority of modified PF. Finally, section \textrm{VI} draws some conclusions out of the proposed works.
\section{Problem statement}
Consider the non-linear dynamic system which can be described by following equations:
\begin{equation} \label{eq:1}
\textit{State equation}\qquad x_k=f_{k-1}(x_{k-1},k-1)+q_{k-1},
\end{equation}
\begin{equation} \label{eq:2}
\textit{Measurement equation}\qquad z_k=h_k(x_k,k)+v_k,
\end{equation}
where $ x_k \in \Re^{n_x} $ denotes the state vector of the system and $ z_k \in \Re^{n_z} $ is the measurement at any discrete time $ k \in (0,1,\cdots) $. $ q_k \in \Re^{n_x} $ and $ v_k \in \Re^{n_z} $ are uncorrelated white noises with arbitrary but known pdf.
Here, we consider the situation that actual measurement received may be a random delayed measurement from previous time steps. This delay can be any steps between $0$ and $N$ for any given $ k^{th}$ instant of time. If any measurement gets delayed by more than $ N$ steps, no measurement is received at estimator and hence the buffer keeps the data received at previous step itself. Here, $N$ is maximum number of admissible delays which is determined carefully as discussed in Subsections \ref{subIIIB} and \ref{subIIIC}.  

To model the delayed measurements at $ k^{th} $ instant, we choose independent and identically distributed Bernoulli random numbers $ \beta_k^i $ $ (i=1,2,\cdots,N+1) $ that take values either 0 or 1 with the unknown probability $ P(\beta_k^i=1)=p=E[\beta_k^i] $ and $ P(\beta_k^i=0)=1-p $, where $ p $ is the unknown latency parameter. If $ y_k $ is the measurement received at $ k^{th} $ instant \cite{singh2017modified}, then

\begin{align} \label{eq:3}
\begin{split} 
y_k&=(1-\beta_k^1)z_k+\beta_k^1(1-\beta_k^2)z_{k-1}+\beta_k^1\beta_k^2(1-\beta_k^3)z_{k-2}+\cdot\\&\cdot \cdot+\prod_{i=1}^N\beta_k^i(1-\beta_{k}^{N+1})z_{k-N}+[1-(1-\beta_k^1)-\beta_k^1(1-\beta_k^2)\\&-\cdots -\prod_{i=1}^N\beta_k^i(1-\beta_k^{N+1})]y_{k-1},\\
&=\sum_{j=0}^N\alpha_k^jz_{k-j} + \left(1-\sum_{j=0}^N\alpha_k^j\right)y_{k-1},
\end{split}
\end{align}
where random variable $ \alpha_k^j$ is defined as
\begin{equation}\label{eq:(4)}
\alpha_k^j=\prod_{i=0}^j\beta_k^i(1-\beta_k^{j+1}).
\end{equation} 
A measurement received at $k^{th}$ time instant, is $j$ step delayed if $\alpha_k^j=1$. Additionally, at any given $k^{th}$ instant of time, at most one of $\alpha_k^j (0\leq j\leq N)$ can be 1. If all $\alpha_k^j$ are zeros that means estimator buffer keeps the measurement received at previous step, \textit{i.e.} $y_{k-1}$. This is a case of measurement loss, as it is delayed by more than $N$ steps.


\textit{Remark 1}:
Bernoulli random variable $\beta_k^i$ and its function $\alpha_k^j$ are immensely practical to represent the real-time randomness of delays in measurements and it is widely used and accepted \cite{ma2011optimal}. Inclusion of the possibility that at a particular step $k$, there is chance that no measurement will be received, increases its practical merit. It can also be observed from \eqref{eq:3} that the same packet can be received more than once at the receiver end, which might be the case for a multi-route network. Moreover, at a given instant of time $k$, $y_k$ includes the noise from only one time step, either one from $v_{k-N:k}$ or the noise received along with $y_{k-1}$.

\textit{Remark 2}:
The latency probability of received measurements $p$, \textit{i.e.} the mean of random variable $\beta_k^i$ is unknown and that is an inevitable real scenario for a practical case. Contrary to the case of single-step delay in \cite{zhang2016particle}, here unknown latency probability is for the arbitrary step delays along with packet drops.

Now, the objective is to outline a PF algorithm for system \eqref{eq:1} with measurement model \eqref{eq:3} which assumes the knowledge of latency probability $p$. Further, $p$ is identified by maximizing the joint probability density of measurements. We propose offline as well as online algorithm to achieve this.

\section{Modified particle filter for randomly delayed measurements}
\subsection{Particle filter}
 As we know, in a sequential importance filter, posterior probability density function is replaced by its equivalent series of weighed particles which can be represented as \cite{arulampalam2002tutorial}
\begin{equation}\label{eq:5}
P(x_{0:k}|z_{1:k})=\sum_{i=1}^{ns} w_k^i\delta[x_{0:k}-x_{0:k}^i],
\end{equation}
where particles $\{x_{0:k}^i\}_{i=1}^{ns}$ are drawn from a proposal density $\textup{q}(x_{0:k}|z_{1:k})$ and weights of particles are chosen using importance principle. Normalized weight of the $i^{th}$ particle can be defined as
\begin{equation}\label{eq:6}
w_k^i=\dfrac{P(x_{0:k}^i|z_{1:k})}{\textup{q}(x_{0:k}^i|z_{1:k})}.
\end{equation}
Now, for a sequential case, we need a recursive weight update at each time step which can be formulated with help of following equations:
\begin{equation}\label{eq:7}
\begin{split}
P(x_{0:k}|z_{1:k})=\dfrac{P(z_k|x_{0:k})P(x_{0:k}|z_{1:k-1})}{P(z_k|z_{1:k-1})}\\
\propto P(z_k|x_k)P(x_k|z_{1:k-1}),\\
\end{split}
\end{equation}
where $P(z_k|z_{1:k-1})$ is a normalizing constant.
Similarly, proposal density can be assumed to be decomposed as
\begin{equation}\label{eq:8}
\textup{q}(x_{0:k}|z_{1:k})=\textup{q}(x_k|x_{0:k-1},z_{1:k})\textup{q}(x_{0:k-1}|z_{1:k-1}).
\end{equation}
 Assuming that the state vector $x_k$ follows the Markov process and if we are interested only in marginal density $P(x_{k}|z_{1:k})$, importance weight of \eqref{eq:6}, with the help of \eqref{eq:7} and \eqref{eq:8}, can be written as
\begin{align} \label{eq:9}
\begin{split}
w_k^i &\propto \dfrac{P(z_k|x_k^i)P(x_k^i|x_{k-1}^i)P(x_{k-1}^i|z_{1:k-1})}{\textup{q}(x_k^i|x_{k-1}^i,z_{1:k-1})\textup{q}(x_{k-1}^i|z_{1:k-1})}\\
&= w_{k-1}^i \dfrac{P(z_k|x_k^i)P(x_k^i|x_{k-1}^i)}{\textup{q}(x_k^i|x_{k-1}^i,z_{1:k})}.
\end{split}
\end{align}
\subsection{Modified PF for randomly delayed measurements} \label{subIIIB}
 A recursive computation for importance weights can be obtained for a nonlinear system with measurement model of \eqref{eq:3}. Before we proceed with computation of modified importance weight, some of the probability values of Bernoulli random variables related with model \eqref{eq:3} need to be obtained.
\begin{lemma} \label{Lemma 1}
	The probability of a received measurement being delayed by $j$ time step, at any instant $k$, is $ P(\alpha_k^j=1)=p^j(1-p)$, $\quad 0\leq j \leq N$.
\end{lemma}
\begin{proof}
	As $\alpha_k^j$ is also a Bernoulli random variable, using its expectation value and \eqref{eq:(4)}
	\begin{align} \label{eq:10}
	\begin{split}
	P(\alpha_k^j=1)=E(\alpha_k^j)
	=E\left(\prod_{i=0}^j\beta_k^i(1-\beta_k^{j+1})\right).
	\end{split}
	\end{align} 
	As $\beta_k^j$ are  i.i.d. Bernoulli random variables, we have
	\begin{align} \label{eq:11}
	\begin{split}
	E\left(\prod_{i=0}^j\beta_k^i(1-\beta_k^{j+1})\right)&=E\left(\prod_{i=0}^j\beta_k^i\right)E(1-\beta_k^{j+1})\\
	&=p^j(1-p).
	\end{split}
	\end{align}  
	Using \eqref{eq:10} and \eqref{eq:11}, we get
	\begin{equation}
	P(\alpha_k^j=1)=p^j(1-p).
	\end{equation}
\end{proof}
\begin{lemma} \label {Lemma 2}
	The probability that estimator will receive $y_{k-1}$ at $k^{th}$ instant of time, is $P(\sum_{j=0}^N \alpha_k^j =0)=p^{N+1}$. 
\end{lemma}
\begin{proof}
	Using basic principle of probability, we can write
	\begin{align}
	\begin{split}
	P\left(\sum_{j=0}^N\alpha_k^j=0\right)&=1-P\left(\sum_{j=0}^N\alpha_k^j=1\right)\\
	&=1-\sum_{j=0}^N P(\alpha_k^j=1).
	\end{split}
	\end{align}
	Using Lemma \ref{Lemma 1}, we get
	\begin{align}\label{eq:14}
	\begin{split}
	P\left(\sum_{j=0}^N\alpha_k^j=0\right)&=1-\sum_{j=0}^N(p^j)(1-p)\\
	&=p^{N+1}.
	\end{split}
	\end{align}
\end{proof}
\textit{Remark 3}:
It can be observed from \eqref{eq:14} that for a high value of $p$, $N$ should be kept sufficiently large to reduce the probability of a packet being lost.

To design the filter for system \eqref{eq:1} and \eqref{eq:3}, we have made some assumptions as follows.
\begin{assumption}
	The state vector $x_{k}$ follows the first order Markov process, \textit{i.e.} $P(x_{k}|x_{1:k-1},y_{1:k})=P(x_{k}|x_{k-1})$, and the received measurement $y_k$, conditionally on $x_{k-N:k}$, is independent of state vectors $x_{1:k-1}$ \textit{i.e.}, $P(y_k|x_{1:k})=P(y_k|x_{k-N:k})$.
\end{assumption}
\begin{assumption}
	There is no correlation between the measurement noises received at two different time steps on the account of random delays and packet loss in \eqref{eq:3}, \textit{i.e.}, $E[v'_jv'_k]_{j\neq k}=0,$ where $v'_j$ is the measurement noise received along with measurement $y_j$ at $j^{th}$ time step. 
\end{assumption}
\begin{lemma}
	Recursion equation of importance weight $w_k^i$ for model \eqref{eq:1} and \eqref{eq:3}, can be obtained as
	\begin{equation}\label{eq:15}
	w_k^i=w_{k-1}^i \dfrac{P(y_k|x_{k-N:k}^i)P(x_k^i|x_{k-1}^i)}{\textup{q}(x_k^i|x_{1:k-1}^i,y_{1:k})},
	\end{equation}
	where $x_k^i$ is drawn from proposal density $\textup{q}(x_k|x_{1:k-1}^i,y_{1:k})$.
\end{lemma}
\begin{proof}
	Using Bayesian theorem, proposal density can be decomposed as
	\begin{align}\label{eq:16}
	\begin{split}
	\textup{q}(x_{0:k}|y_{1:k}) &=\textup{q}(x_k|x_{1:k-1},y_{1:k})\textup{q}(x_{1:k-1}|y_{1:k})\\
	&=\textup{q}(x_k|x_{1:k-1},y_{1:k})\textup{q}(x_{1:k-1}|y_{1:k-1}). 
	\end{split}
	\end{align}
	 Particles $x_k^i$ and $x_{1:k-1}^i$ can be sampled from $\textup{q}(x_k|x_{1:k-1},y_{1:k})$ and $\textup{q}(x_{1:k-1}|y_{1:k-1})$, respectively. Again, using Bayesian rule, joint pdf, $P(x_{1:k},y_{1:k})$, can be considered to be decomposed as below.
	\begin{align}\label{eq:17}
	\begin{split}
	P(x_{1:k},&y_{1:k})\\
	&=P(y_k|x_k,x_{1:k-1},y_{1:k-1})P(x_k,x_{1:k-1},y_{1:k-1})\\
	&=P(y_k|x_k,x_{1:k-1},y_{1:k-1})P(x_k|x_{1:k-1},y_{1:k-1})\\& \times P(x_{1:k-1},y_{1:k-1}).
	\end{split}
	\end{align}
	By Assumption 1, \eqref{eq:17} can be rewritten as 
	\begin{equation}\label{eq:18}
	P(x_{1:k},y_{1:k})=P(y_k|x_{k-N:k})P(x_k|x_{k-1})P(x_{1:k-1},y_{1:k-1}).
	\end{equation}
	Using \eqref{eq:16} and \eqref{eq:18}, the importance weight can be written as
	\begin{align}\label{eq:19}
	\begin{split}
	w_k 
	&=\dfrac{P(y_k|x_{k-N:k})P(x_k|x_{1:k-1})}{\textup{q}(x_k|x_{1:k-1},y_{1:k})} \dfrac{P(x_{1:k-1},y_{1:k-1})}{\textup{q}(x_{1:k-1}|y_{1:k-1})}\\
	&=w_{k-1}\dfrac{P(y_k|x_{k-N:k})P(x_k|x_{1:k-1})}{\textup{q}(x_k|x_{1:k-1},y_{1:k})}.
	\end{split} 
	\end{align}
	Now, with the help of  \eqref{eq:19}, $w_k^i$ can be finally written as \eqref{eq:15}.
	\end{proof}
	
	\textit{Note:}  Alternatively, Lemma 3 can be proved by augmenting state vector $x_k$ into $\bar{x_k} =[x_k,x_{k-1},\cdots,x_{k-N}]^T$ as new state vector of the same system. Using Assumptions 1 and 2, Eq. \eqref{eq:9} when expressed in terms of state vector $\bar{x}_k^i$ and received measurement $y_k$ in place of $x_k^i$ and $z_k$ respectively, can be simplified as \eqref{eq:15}.
	
	\begin{lemma} Likelihood density $P(y_k|x_{k-N:k}^i)$ can be computed recursively as
	\begin{align}\label{eq:20}
	\begin{split}
	P(y_k|x_{k-N:k}^i)=\sum_{j=0}^N p^j(1-p)P_{v_{k-j}}(y_k-h_{k-j}(x_{k-j}^i))\\ +
	p^{N+1} P(y_{k-1}|x_{k-1-N:k-1}^i),
	\end {split}
	\end{align}
	where $P_{v_{k-j}}(\boldsymbol{\cdot})$ represents the pdf of measurement noise $v_{k-j}$.
\end{lemma}
\begin{proof}
	In \eqref{eq:19},  the measurement $y_k$ may be correlated with $x_{k-j}\ (j=0,1,\cdots,N)$ if delay occurs.
	Let $\gamma_k$ be a Bernoulli random variable that denotes a measurement has been received (with any number of step delay between 0 and $N$). Now, using Lemma \ref{Lemma 2}, probability that a measurement has been received with admissible delay, is
	\begin{equation}\label{eq:21}
	P(\gamma_k=1)=\sum_{j=0}^N p^j(1-p).
	\end{equation} Consequently, probability that no measurement (\textit{i.e.}, no of delay step is more than $N$ or any other factor for packet loss ) has been received, is 
	\begin{equation}\label{eq:22}
	P(\gamma_k=0)=1-\sum_{j=0}^N p^j(1-p)=p^{N+1}.
	\end{equation}
	Assuming no correlation between $\gamma_k$ and the state vectors, marginalization of likelihood density from its joint density can be evaluated as
	\begin{align}
	\begin{split}
	P(y_k|&x_{k-N:k})\\&=\sum_{\gamma_k=0}^1 P(y_k,\gamma_k|x_{k-N:k})\\
	&=\sum_{\gamma_k=0}^1 P(y_k|\gamma_k,x_{k-N:k}) P(\gamma_k|x_{k-N:k})\\
	&=\sum_{\gamma_k=0}^1 P(y_k|\gamma_k,x_{k-N:k})P(\gamma_k)\\
	&= P(y_k|\gamma_k=0,x_{k-n:k})P(\gamma_k=0)\\&+P(y_k|\gamma_k=1,x_{k-N:k})P(\gamma_k=1).
	\end{split}
	\end{align} 
	Again, using \eqref{eq:21} and \eqref{eq:22}, we have
	\begin{equation}\label{eq:24}
	\begin{split}
	P(y_k|x_{k-N:k})&=\sum_{j=0}^N( p^j(1-p)P(y_k|\gamma_k=1,x_{k-N:k}))+\\ &p^{N+1}P(y_k|\gamma_k=0,x_{k-N:k}).
	\end{split}
	\end{equation}
	Now, we can expand \eqref{eq:3} using \eqref{eq:2} and write as
	\begin{equation}\label{eq:25}
	y_k=\sum_{j=0}^N(\alpha_k^j(h_{k-j}(x_{k-j})+v_{k-j})) +(1-\sum_{j=0}^N\alpha_k^j)y_{k-1}.
	\end{equation}
	If $\gamma_k=1$ (\textit{i.e.}, one of $\alpha_k^j(0\leq j \leq N)$ is 1),  \eqref{eq:25} can be rewritten as $y_k= h_{k-j}(x_{k-j})+v_{k-j}$ and consequently,
	\begin{equation}\label{eq:26}
	P(y_k|\gamma_k=1,x_{k-N:k})=P_{v_{k-j}}(y_k-h_{k-j}(x_{k-j})).
	\end{equation}
	Similarly, if $\gamma_k=0$ (\textit{i.e.}, all of $\alpha_k^j(0\leq j \leq N)$ are 0), \eqref{eq:25} can be rewritten as $y_k=y_{k-1}$ and
	\begin{equation}\label{eq:27}
	P(y_k|\gamma_k=0,x_{k-N:k})=P(y_{k-1}|x_{k-1-N:k-1}).
	\end{equation}
	Now, substituting \eqref{eq:26} and \eqref{eq:27} into \eqref{eq:24} and writing it for the particles (\textit{i.e.}, $P(y_k|x_{k-N:k}^i)$),  we can easily obtain \eqref{eq:20}.
\end{proof}
\textit{Remark 4}:
It can be observed that if there are no random delays and packet drops in the received measurements (\textit{i.e}, $N=0$ and $p=0$),  \eqref{eq:20} converges to a likelihood density $P_{v_k}(z_k-h_k(x_k^i))$ of a standard PF.
 
In this work, to reduce the effect of degeneracy in iterative updating of importance weight of particles, resampling is chosen to be done at each step. Generally, we choose prior density $P(x_k|x_{k-1})$ as a proposal density $q(x_k|x_{1:k-1},y_{1:k})$ to implement sequential importance resampling (SIR) PF \cite{arulampalam2002tutorial}. However, posterior density obtained from a standard non-linear filter can also be used as proposal density function.
\begin{algorithm} 
	\caption{Modified Particle Filter}\label{algo:1}
	\begin{center}
		$[\{x_k^i,w_k^i\}_{i=1}^{ns}]:= \texttt{SIR}[\{x_{k-N:k}^i,w_{k-1}^i\}_{i=1}^{ns}, \hat{p},y_k]$
	\end{center}
	\begin{itemize}
		\item \textit{for} $i=1: ns$
		\begin{itemize}
			\item \texttt{Draw} $x_k^i \sim P(x_k|x_{k-1}^i)$
			\item \texttt{Compute the importance weight $w_k^i$ according to equations \eqref{eq:15} and \eqref{eq:20}}
			\item \texttt{Normalize the importance weight}:$\; w_k^i:=w_k^i/\texttt{SUM}[\{w_k^i\}_{i=1}^{ns}]$
		\end{itemize}
		\item \textit{end for}
		\item \texttt{Resample the particles at each step} 
		\begin{itemize}
			\item $[\{x_k^i, w_k^i\}_{i=1}^{ns}]:=\texttt{RESAMPLE}[\{x_k^i, w_k^i\}_{i=1}^{ns}]$
		\end{itemize}
	\end{itemize}
\end{algorithm}


\subsection{Convergence of the PF for randomly delayed measurements}\label{subIIIC}
In this subsection, we are going to explore the conditions for convergence of the modified PF derived for randomly delayed measurements. A PF can be said to be converging if its empirical approximation admits a mean square error of order $1/ns$ at each step of filtering \cite{crisan2002survey}. According to \cite{crisan2002survey}, prime requisite for simple convergence is that likelihood function $P(z_k|\boldsymbol{\cdot})$ should be bounded, \textit{i.e.} $\|P(z_k|\boldsymbol{\cdot)}\|<\infty$, for all $x_k \in \Re^{n_x}$. Following lemma will make sure that this requisite holds for our case.
\begin{lemma}
	If, for non-delayed measurement $z_k$ 
	\begin{equation}
	\|P(z_k|x_k)\| < \infty,\quad \forall  x_k \in \Re^{n_x},
	\end{equation}
	then
	\begin{equation}\label{eq:29}
	 \|P(y_k|x_{k-N:k})\|< \infty,\quad \forall x_k \in \Re^{n_x}.
	\end{equation}
\end{lemma}
\begin{proof}
	Using \eqref{eq:2}, we can write
	\begin{equation*}
	P(z_k|x_k)=P_{v_k}(z_k-h_k(x_k)).
	\end{equation*}	
	Thus, pdf of white noise $v_k$, \textit{i.e.} $P_{v_k}(z_k-h_k(x_k))$ is bounded for all its real-valued inputs. Now, rearranging terms of Eq. \eqref{eq:20} on both side, we have
	\begin{equation}\label{eq:30}
	\begin{split}
	P(y_k|x_{k-N:k})-p^{N+1}P(y_{k-1}|x_{k-1-N:k-1})\\=\sum_{j=0}^N p^j(1-p)
	P_{v_{k-j}}(y_k-h_{k-j}(x_{k-j})).
	\end{split}
	\end{equation}
Now, being a white noise, $v_k$ is a stationary process. Hence, its pdf is not affected by time shift, \textit{i.e.} if $P_{v_k}(\boldsymbol{\cdot})$ is bounded, $P_{v_{k-j}}(\boldsymbol{\cdot})$ must be bounded. Again, since $p^j(1-p)<1$ for all values of $j$, we can write
	\begin{equation}\label{eq:31}
	\begin{split}
	\sum_{j=0}^N p^j(1-p)P_{v_{k-j}}(y_k-h_{k-j}(x_{k-j}))\\ \leq N P_{v_{k-j}}(y_k-h_{k-j}(x_{k-j})).
	\end{split}
	\end{equation}
	Given $N$ is a finite number, \eqref{eq:30} and \eqref{eq:31} can easily be used to establish \eqref{eq:29}.
\end{proof}

Now, Theorem 2 of \cite{crisan2002survey}, for any $\Phi \in B(\Re^{n_x})$, where $B$ denotes the Borel set, can be expressed as
\begin{equation*}
\begin{split}
E[((P^{ns}(y_k|x_{k-N:k}),\Phi)&-(P(y_k|x_{k-N:k}),\Phi))^2]\\& \leq c_{k|k}\dfrac{ \|\Phi\|^2}{ns},
\end{split}
\end{equation*}
where $(\nu, \phi)$ is defined as \begin{equation*}
(\nu, \phi)=\int \phi\nu.
\end{equation*}
 Here, $\nu $ is a probability density and $\phi$ is a Borel set. $P^{ns}(y_k|x_{k-N:k})$ is empirical approximation of $P(y_k|x_{k-N:k})$ given by \eqref{eq:20} using $ns$ support points. Theorem 3 of \cite{crisan2002survey} suggests that $c_{k|k}$ is a constant which is independent of number of particles $ns$, but represents the dependency of mixed dynamics of system on initial conditions or past values. That is, if optimal filter associated with the dynamics of system has long memory, $c_{k|k}$ will go on accumulating the mean square error with each step. Unfortunately, in our case, the filter does have long memory due to arbitrary delays in measurements. Therefore, in order to avoid the large mean square error of convergence, value of $N$ should be kept small. On the other hand, to counter the increase in value of $c_{k|k}$, number of particles $ns$ needs to be large.

\textit{Remark 5}:
By Lemma \ref{Lemma 2}, we observe that to reduce the possibility of information loss, $N$ should be large if $p$ is high. But, from the above discussion, high value of  $N$ can lead to large convergence error. Thus, we need to maintain a balance between these two situations in selecting a value of $N$ if randomness in delay is more likely.

\section{Identification of latency probability}
In practice, when a set of randomly delayed measurements is given, we may not have knowledge about channel properties and its random parameters. Hence, latency probability which is required to design the filter can be unknown. Here, we present ML criterion to identify the unknown latency probability for received measurements. This method involves maximization of the joint density $P_p(y_{1:m})$ of received measurements that is a function of latency parameter $p$, which can mathematically be represented as \cite{zhang2016particle}
\begin{equation}\label{eq:32}
\hat{p}=arg\ \underset{p\in[0,1]}{max} P_p(y_1,\cdots,y_m),
\end{equation} 
where $m$ is number of measurements taken for identification of parameter $p$ and $\hat{p}$ is estimated value of $p$.
Now, we assume that the first received measurement $y_1$ is independent of parameter $p$ and is equal to $z_1$. Again, using Bayesian theorem the above joint pdf can be reformulated as  
\begin{equation}\label{eq:33}
P_p(y_1,\cdots,y_m)=P(y_1)\prod_{k=2}^m P_p(y_k|y_{1:k-1}).
\end{equation}  

Sometimes, for computational advantage above maximization of likelihood is expressed in terms of log-likelihood (LL), as log is a monotonic function it does not affect the process. The LL of \eqref{eq:33} can be formulated as
\begin{align}\label{eq:34}
\begin{split}
L_p(y_{1:m})&=\log P_p(y_{1:m})\\
&=\log P(y_1)+\sum_{k=2}^m\log P_p(y_k|y_{1:k-1}),
\end{split}
\end{align} 
where $L_p(y_{1:m})$ is LL function for the received measurements. Now, to solve the maximization  problem of \eqref{eq:32}, two things need to be done in order. First, computation of likelihood $P_p(y_k|y_{1:k-1})$ and other, the maximization of $L_p(y_{1:m})$.
\subsection{Computation of likelihood density}
\begin{lemma}
	The likelihood density function $P_p(y_k|y_{1:k-1})$ of randomly delayed measurements can be approximated as
	\begin{equation}\label{eq:35}
	P_p(y_k|y_{1:k-1})=\dfrac{1}{ns}\sum_{i=1}^{ns}P_p(y_k|x_{k-N:k}^i),
	\end{equation}
	where $P_p(y_k|x_{k-N:k}^i)$ is given in \eqref{eq:20}.
\end{lemma}
\begin{proof}
	We can express the likelihood density $P_p(y_k|y_{1:k-1})$ as marginal density of a joint pdf that includes delayed measurement and previous states that are correlated, using Bayesian theorem and total probability, that is
	\begin{align} \label{eq:36}
	\begin{split}
	&P_p(y_k|y_{1:k-1})\\
	&=\idotsint P_p(y_k,x_{k-N:k}|y_{1:k-1})dx_k\cdots dx_{k-N}\\
	&=\idotsint P_p(y_k|x_{k-N:k},y_{1:k-1}) P_p(x_{k-N:k}|y_{1:k-1})dx_k\cdots \\&\times dx_{k-N}\\
	&=\idotsint P_p(y_k|x_{k-N:k})P_p(x_{k-N:k}|y_{1:k-1})dx_k\cdots dx_{k-N}.
	\end{split}
	\end{align}
	Again, using Bayes' rule we can write
	\begin{align*} 
	\begin{split}
	P_p(x_{k-N:k}|y_{1:k-1})=P_p(x_k|x_{k-N:k-1},y_{1:k-1})P_p(x_{k-N:k-1})\\
	=P_p(x_k|x_{k-N:k-1},y_{1:k-1})P_p(x_{k-1}|x_{k-N:k-2},y_{1:k-1})\times\\\cdots \times P_p(x_{k-N}|y_{1:k-1}).
	\end{split}
	\end{align*}
	As state prior density functions are independent of measurements, they are not function of latency probability and again, by using Assumption 1 and \eqref{eq:5}, we can write joint pdf
	\begin{align} \label{eq:37}
	\begin{split}
	&P_p(x_{k-N:k}|y_{1:k-1})=P(x_k|x_{k-1})P(x_{k-1}|x_{k-2})\cdots \\&\times P_p(x_{k-N}|y_{1:k-N}),\\
	&\approx \dfrac{1}{ns}\sum_{i=1}^{ns}\delta[x_k-x_k^i]\delta[x_{k-1}-x_{k-1}^i]\cdots\delta[x_{k-N}-x_{k-N}^i],
	\end{split}
	\end{align}
	where $x_{k-N}^i,x_{k-N+1}^i,\cdots,$ and $x_k^i$ are drawn from $P_p(x_{k-N}|y_{1:k-N}),P(x_{k-N+1}|x_{k-N}),\cdots,$ and $P(x_k|x_{k-1})$ respectively, and normalized importance weight of particles are $1/ns$. It is to be noted that particles being drawn from prior density are independent of measurements and hence from latency probability. Now, substitute \eqref{eq:37} into \eqref{eq:36} and $P_p(y_k|y_{1:k-1})$ can be approximately computed as
	\begin{equation*}
	P_p(y_k|y_{1:k-1})=\dfrac{1}{ns}\sum_{i=1}^{ns}P_p(y_k|x_k^i,\cdots,x_{k-N}^i),
	\end{equation*}
	where $P_p(y_k|x_{k-N:k}^i)$ is given in \eqref{eq:20}.
\end{proof}
Algorithm \ref{algo:2} illustrates the steps for computation of LL function.
\subsection{Maximization of log-likelihood function}
Substituting \eqref{eq:35} into \eqref{eq:34}, we can rewrite LL function \eqref{eq:34} as
\begin{equation}\label{eq:38}
L_p(y_{1:m})=\log P(y_1)+\sum_{k=2}^{m}\log\left(\dfrac{1}{ns}\sum_{i=1}^{ns}P_p(y_k|x_{k-N:k}^i)\right).
\end{equation}
$y_1$ is independent of parameter $p$ and can be neglected for the purpose of  maximization of likelihood density. Moreover, if proposal density $\textup{q}(x_k|x_{1:k-1}^i,y_{1:k})$ is prediction density function $P(x_k^i|x_{k-1}^i)$, then using \eqref{eq:15} we can rewrite the utility function \eqref{eq:38} as
\begin{equation}\label{eq:39}
L_p(y_{1:m})=\sum_{k=2}^m\log\left[\sum_{i=1}^{ns}w_{p,k}^i\right].
\end{equation}
Eq. \eqref{eq:39} can easily be maximized numerically over $p \in[0,1]$.

There are two ways to go for this numerical search of latency  parameter: offline and online identification. In the offline method, we can use more measurements for higher accuracy of parameter estimation. If in \eqref{eq:39}, value of $m$ is really high ($m\rightarrow \infty$) and parameter $p$ is varied in  very small increments ($sl$$ \rightarrow 0$) then, estimated latency probability can converge towards its true value. In practice, we may not afford to have that many measurements. Besides, it also will increase the computational burden. As identification of latency parameter is done only once at the beginning of filtering, computation time of offline mode is limited. Offline identification can be done only after receiving $m$ observations and algorithm should be run for multiple times for improved accuracy of estimated value. Algorithm \ref{algo:3} dictates the steps for offline identification.

In case of online identification, as we are estimating the latency parameter at each time step, initially we have too little information to extract the latency probability through maximization and hence very poor accuracy. However, accuracy of parameter value becomes comparable with that of offline identification with increase in time steps. The Advantage with this method is it provides latency probability at each step of state estimation for less computational effort, unlike offline method where identification is done after considering many measurements. Running average of estimated values at each step can be evaluated for improved accuracy of identified parameter. Algorithm \ref{algo:4} outlines this method of identification. 
\begin{algorithm} 
	\caption{Computation of log-likelihood function}\label{algo:2}
	\begin{center}
		[$L_p,\{x_k^i,w_{p,k}^i\}_{i=1}^{ns}]:= \texttt{LL}[L_p,\{x_k^i,w_{p,k}^i\}_{i=1}^{ns},y_k$]
	\end{center}
	\begin{itemize}
		\item \textit{for} ($i=1:ns$)
		\begin{itemize}
		\item \texttt{Draw $x_k^i \sim P(x_k|x_{k-1}^i)$}
		\item \texttt{Assign particle a weight}:
		\begin{equation*}
		\begin{split}
		w_{p,k}^i=\dfrac{1}{ns}(\sum_{j=0}^N (p^j(1-p)P_{v_{k-j}}(y_k-h_{k-j}(x_{k-j}^i))) +\\p^{N+1} P(y_{k-1}|x_{k-1-N:k-1}^i))
		\end{split}
		\end{equation*}
		\end{itemize}
		\item \textit{end for}
		\item \texttt{Compute the LL function}:
		\begin{equation*}
		L_p:=L_p+\log (\dfrac{1}{ns}\sum_{i=1}^{ns}w_{p,k}^i)
		\end{equation*}
		\item \textit{for} ($i=1:ns$)
		\begin{itemize}
		\item \texttt{Normalize the importance weight}: $w_{p,k}^i:=w_{p,k}^i/\texttt{SUM}[\{w_{p,k}^i\}_{i=1}^{ns}]$
		\end{itemize}
		\item \textit{end for}
		\item \texttt{Resample the drawn particles at each step}:
		\begin{enumerate}[]
		\item $[\{x_k^j,w_{p,k}^j\}_{j=1}^{ns}]:=\texttt{RESAMPLE}[\{x_k^i,w_{p,k}^i\}_{i=1}^{ns}]$
		\end{enumerate}
		
	\end{itemize}
\end{algorithm}

\begin{algorithm}
		\caption{Offline identification of latency probability}\label{algo:3}
		\begin{center}
			[$L_{max},\hat{p}]:=\texttt{OFFLINE}[L_p,p,sl,m,\{x_k^i,w_k^i\}_{i=1}^{ns}$]
		\end{center}
		\begin{itemize}
			\item \texttt{Select the values for $sl$ and $m$ }
			\item \texttt{Set} $L_{max}=0$ \texttt{and} $\hat{p}=0$
			\item \textit{for} ($p=0:sl:1$)
			\begin{itemize}
				\item \texttt{Initialize the log-likelihood function with} $L_p=0$ 
				\item \textit{for} ($k=1:m$)
				\begin{enumerate}[--]
					\item $[L_p,\{x_k^i,w_{p,k}^i\}_{i=1}^{ns}]=\texttt{LL}[L_p,\{x_k^i,w_{p,k}^i\}_{i=1}^{ns},y_k]$
				\end{enumerate}
				\item \textit{end for}
				\item \texttt{Update:} \texttt{If} $L_p>L_{max}$ \texttt{then}
				\begin{enumerate}[]
					\item $L_{max}=L_p$ \texttt{and} $\hat{p}=p$
				\end{enumerate} 
			\end{itemize}
			\item \textit{end for}
		\end{itemize}
\end{algorithm}
\begin{algorithm}
	\caption{Online identification of latency probability}\label{algo:4}
	\begin{center}
		[$L_{max},\hat{p}]:=\texttt{ONLINE}[L_p,p,sl,\{x_k^i,w_k^i\}_{i=1}^{ns}$]
	\end{center}
	\begin{itemize}
		\item \texttt{Select the value for $sl$ and set $L_p=0$}
		\item \textit{for} $(t=1:k)$
		\begin{itemize}
			\item \texttt{Set} $L_{max}=0$  \texttt{and} $ \hat{p}=0$ 
			\item \textit{for} $(p=0:sl:1)$
			\begin{enumerate}[--]
				\item $[L_p,\{x_t^i,w_{p,t}^i\}_{i=1}^{ns}]:= \texttt{LL}[L_p,\{x_t^i,w_{p,t}^i\}_{i=1}^{ns},y_t]$
			\end{enumerate}
			\item \textit{end for}
			\item \texttt{Update:} \texttt{If} $L_p>L_{max}$ \texttt{then}
			\begin{enumerate}[]
				\item $L_{max}=L_p$ \texttt{and} $\hat{p}=p$
			\end{enumerate}
		\end{itemize}
		\item \textit{end for}
	\end{itemize}
\end{algorithm}
\section{Simulation results}

 To showcase the effectiveness of proposed  PF against the established standard PF for a set of likely delayed measurements, two different types of filtering problems are simulated. In this simulation work, latency probability of delayed measurements is first identified by the numerical maximization of \eqref{eq:39} over $p \in[0,1]$, with the help of proposed PF algorithm. Subsequently, the estimated probability $p$  is used to implement the proposed PF for the given problems. A filter with the same structure as the proposed PF but with other than the true value of $N$ is also implemented to investigate the impact of selecting wrong value of maximum admissible delay. Finally, performance of the all filters are compared in terms of RMSE. It is also to be noted that the proposed PF with wrong value of $N$ are implemented with the true value of latency probability.  
 \subsection{Problem 1}
 
Non-stationary growth model has been widely used in literature for validation of performances by nonlinear filters\cite{hermoso2007extended,hermoso2009unscented,kotecha2003gaussian,wang2014design}. Its non-linear dynamics can be represented by following equations
\begin{equation}
x_k=0.5x_{k-1} + 25\dfrac{x_{k-1}}{1+x_{k-1}^2} +8 \cos(1.2k)+q_{k-1},
\end{equation} 
\begin{equation}\label{eq:41}
z_k=\dfrac{x_k^2}{20} +v_k,
\end{equation}
where $q_k$ and $v_k$ are uncorrelated white processes. For this simulation work, state and measurement noises are  considered as zero mean Gaussian with $E[q_k^2]=10$ and $E[v_k^2]=1$, respectively. True state $x_0$ is also taken as Gaussian random variable having zero mean and unity variance. The received measurements data, $y_{1:k}$ are generated by using \eqref{eq:41} and  \eqref{eq:3} with $N=2$. The number of particles used for simulation of this problem is, $ns=1000$.

Offline identification of latency probability is carried out by using Algorithm \ref{algo:3} with $sl=0.01$ and $m=500$. Latency probability ($p$) at the end of each ensemble is calculated and plotted in Figure \ref{fig_1}. For this case, when true value of $p$ is $0.5$, mean of that over 100 ensembles is calculated as $0.481$.  Online identification of latency probability is shown in Figure \ref{fig_2}. Here, identified probability at each time step is running average of estimated probabilities. We can see from Figure \ref{fig_2} that as time increases, $p$ converges close to its true value.

For filtering, proposed PF is implemented with $N=1$  and $N=2$ for same set of received measurements and results are compared against that of standard PF. To compare the results, root mean square errors (RMSE) calculated over 100 Monte Carlo (MC) runs are plotted over 50 time steps in Figure \ref{fig_3}. It can been seen that proposed PF with $N=2$ outperforms the other two filters. Moreover, it is interesting to observe that performance of proposed PF with $N=1$ is better than that of standard PF. Further, the average RMSE calculated over $50$ time steps for different values of true latency probability is shown in Figure \ref{fig_4}. As expected, at higher probability value where  packet drop is more likely, filter designed for $N=2$ performs better than the other two.     
\begin{figure}[!t]
\centering
\includegraphics[width=3in,height=2.5in]{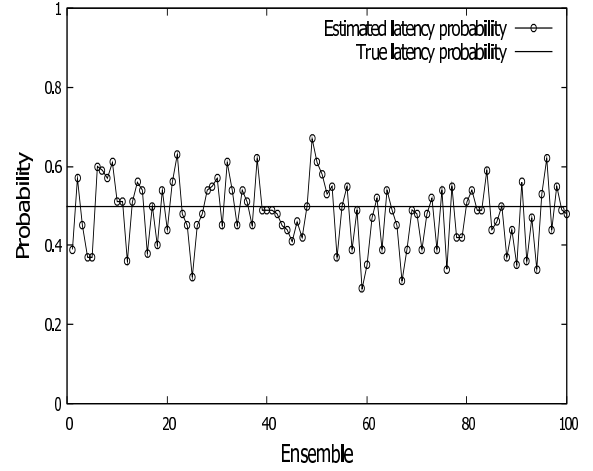}
\caption{Offline estimated latency probability (Problem 1)}
\label{fig_1}
\end{figure}
\begin{figure}[!t]
	\centering
\includegraphics[width=3in,height=2.5in]{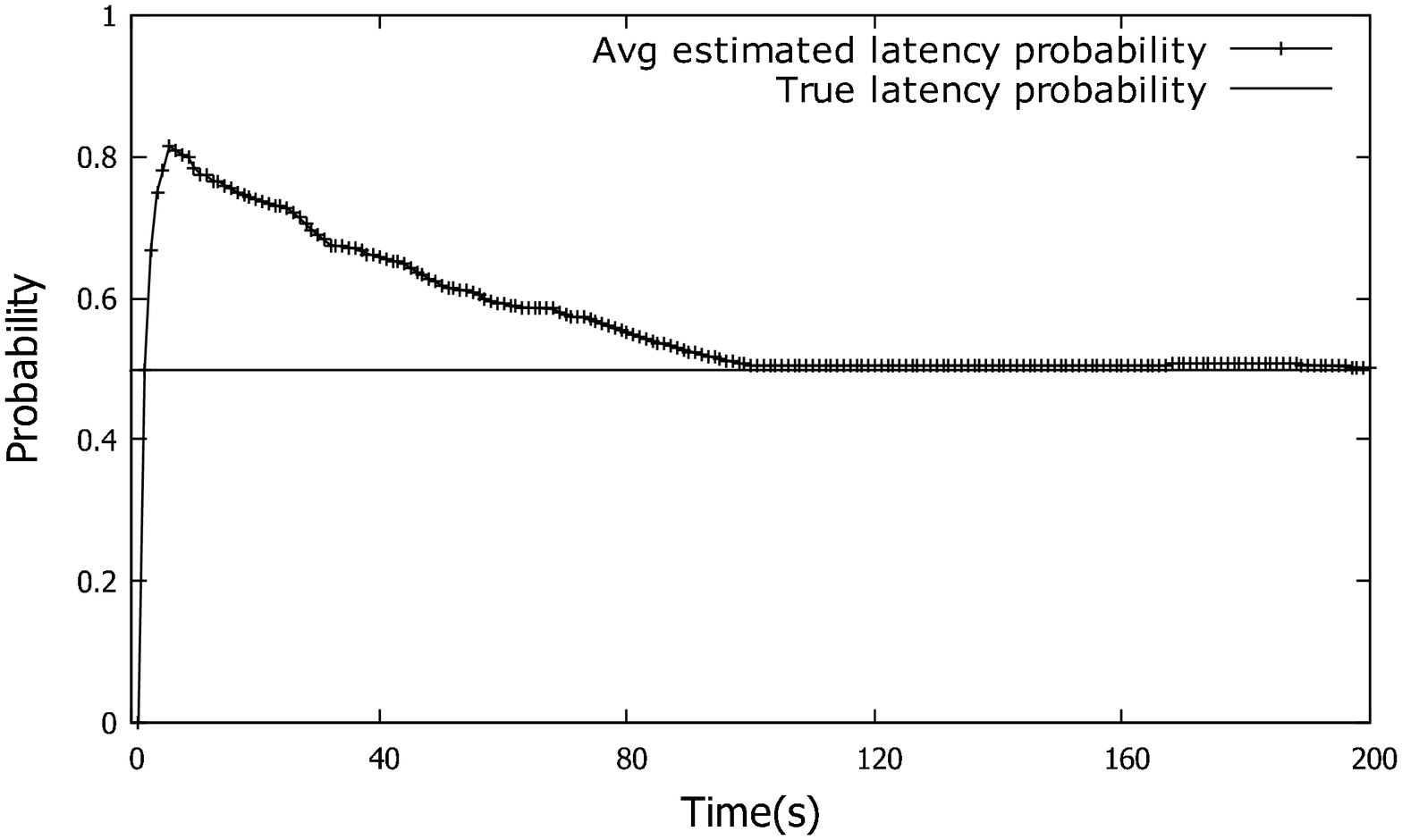}
\caption{Online estimated latency probability (Problem 1)}
\label{fig_2}
\end{figure}
\begin{figure}[!t]
	\centering
	\includegraphics[width=3in,height=2.5in]{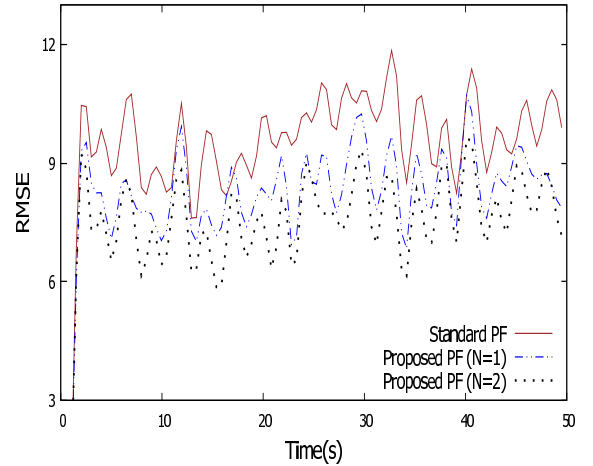}
	\caption{RMSE vs time for different filters considering $\hat{p}=0.481$ (Problem 1)}
	\label{fig_3}
\end{figure}
\begin{figure}[!t]
	\centering
	\includegraphics[width=3in,height=2.5in]{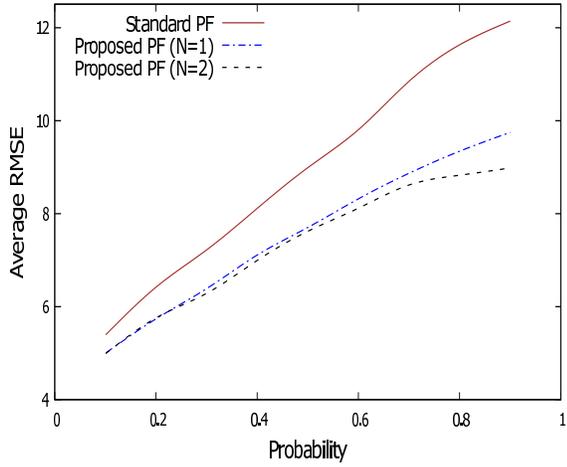}
	\caption{Average RMSE vs probability for different filters (Problem 1)}
	\label{fig_4}
\end{figure}
\subsection{Problem 2}
 
 In this simulation work, we consider a bearing-only tracking (BOT) problem where a moving target is being tracked from a moving platform \cite{lin2002comparison}. The BOT problem has mainly two components, namely, the target kinematics and the tracking platform kinematics as shown in Figure \ref{fig_5}. The tracking platform motion may be represented by the following equations
 \begin{align}
 \begin{split}
 x_{tp,k}&=\bar{x}_{tp,k} +\varDelta x_{tp,k}, \qquad k=0,1,\cdots,n_{step}\\
 y_{tp,k}&=\bar{y}_{tp,k}+ \varDelta y_{tp,k}, \qquad k=0,1,\cdots,n_{step},
 \end{split}
 \end{align}
 where $x_{tp,k}$ and $y_{tp,k}$ represent the X and Y co-ordinates of tracking platform at $k^{th}$ time-step, respectively.  $\bar{x}_{tp,k}$ and $\bar{y}_{tp,k}$ are known mean co-ordinates of platform position and $\varDelta x_{tp,k}$ and $\varDelta y_{tp,k}$ are uncorrelated zero mean Gaussian white noise with variances, $r_x=1  m^2$ and $r_y=1 m^2$, respectively. The mean values for position co-ordinates (in meters) are $\bar{x}_{tp,k}=4kT$ and $\bar{y}_{tp,k}=20$, where $T$ is sampling time for discretization expressed in \textit{seconds} ($s$).
 
 The target moves in X direction according to following discrete state space relations.
 \begin{figure}[]
 	\centering
 	\includegraphics[width=3in,height=2.5in]{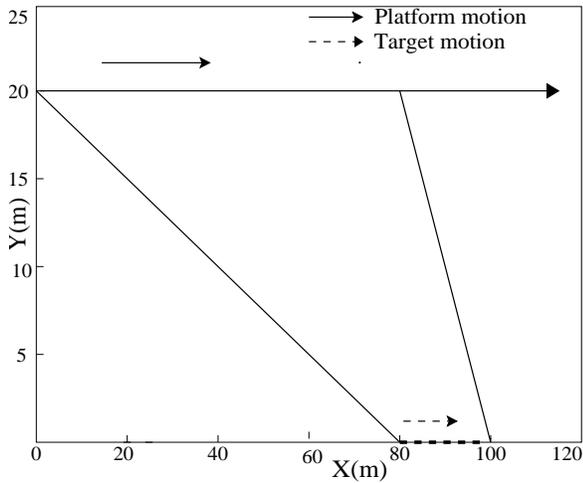}
 	\caption{Platform and target trajectory}
 	\label{fig_5}
 	\end{figure}
 \begin{equation}
 x_{k}=Fx_{k-1} +Gq_{k-1},
 \end{equation}
 where
 $$x_k=\begin{bmatrix} 
 	x_{1,k}\\
 	x_{2,k}
 \end{bmatrix}, \quad F=\begin{bmatrix}
 1 & T\\
 0 & T
 \end{bmatrix},\quad G=\begin{bmatrix}
 T^2/2\\
 T
 \end{bmatrix}$$
 with $x_{1,k}$ denoting the position along X axis and $x_{2,k}$ denoting the velocity (in \textit {m/s}) of the target. $q_k$ is independent zero mean white Gaussian noise with variance $r_q=0.01\;m^2/s^4$ and initial true states are assumed to be $ x_0=[80\;1]^T$.
 
 The sensor measurement is given by
 \begin{equation}\label{eq:44}
 z_{k}=z_{m,k}+v_k,
 \end{equation}
 where
 \begin{equation*}
 z_{m,k} = h[x_{tp,k},y_{tp,k},x_{1,k}]=\arctan \dfrac{y_{tp,k}}{x_{1,k}-x_{tp,k}}
 \end{equation*}
 is the angle between the X axis and the line of sight from the sensor to target and $v_k$ is Gaussian with zero mean and variance $r_v=(3^{\circ})^2$, which is assumed to be independent of platform motion noises.
 
For the simulation work of this BOT problem, measurements data $y_{1:k}$ are generated for $21\, s$ by using \eqref{eq:3} and \eqref{eq:44} with $N=2$. Number of particles used for approximating the pdf is, $ns=3000$.  At the beginning of simulation, latency probability of received measurements is identified offline with $T=0.05\,s$, $m=400$, and $sl=0.01$. Latency probability at the end of each ensemble is calculated and plotted in Figure \ref{fig_6}. Mean value of estimated $p$ over 100 ensembles is calculated as $0.460$, whereas  its true value is $0.5$. For online identification, $T$ is taken as $0.1\;s$ and running average of latency probability is calculated at each time step. Figure \ref{fig_7} shows the identified latency probability for $200$ time steps. From the plot, it can be seen that as number of time steps increase, $p$ converges close to its true value.
\begin{figure}[]
	\centering
	\includegraphics[width=3in,height=2.5in]{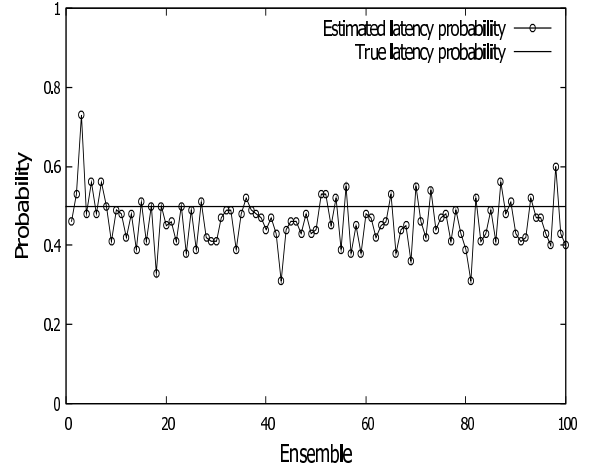}
	\caption{Offline estimated latency probability (Problem 2)}
	\label{fig_6}
\end{figure}
\begin{figure}[]
	\centering
	\includegraphics[width=3in,height=2.5in]{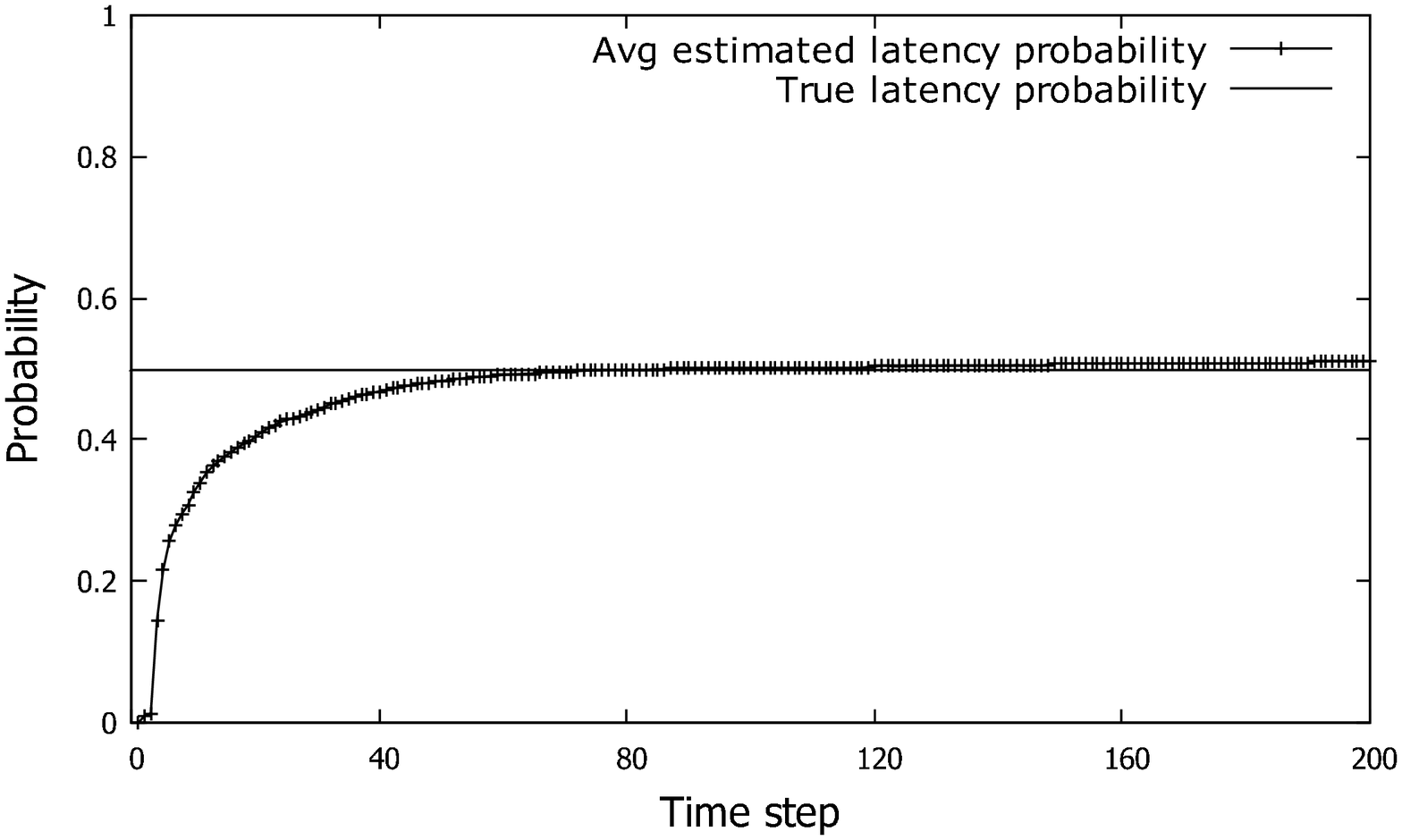}
	\caption{Online estimated latency probability (Problem 2)}
	\label{fig_7}
\end{figure}
\begin{figure}[!t]
	\centering
	\includegraphics[width=3in,height=2.5in]{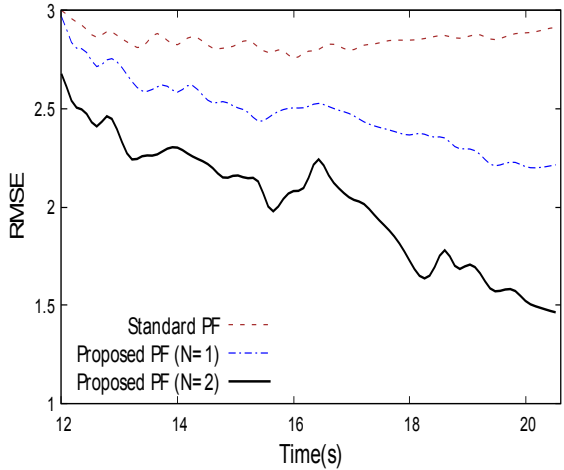}
	\caption{RMSE vs time for state $x_{1,k}$ considering $\hat{p}=0.46$ (Problem 2)}
	\label{fig_8}
	\end{figure}

Now, the proposed PF is implemented with $N=1$ and $N=2$ and corresponding estimated latency probabilities are used for filtering. To show effectiveness, the results of standard PF for same set of delayed measurements are compared against that of proposed PF. RMSE of three filters with true latency probability $p=0.5$ and sampling time $T=0.2\,s$, which have been calculated over $100$ MC runs, are plotted in Figures \ref{fig_8}-\ref{fig_9}. From the plots, it can be concluded that using a filter designed for delayed measurements is a better choice than a conventional PF where delays are not accounted.

Further in this simulation, the average RMSE is calculated over $100$ time steps for different values of $p$. The average RMSE for two states,  $x_{1,k}$ and $x_{2,k}$ are plotted in Figures \ref{fig_10} and \ref{fig_11} respectively. It can be seen that the difference in performance of the proposed PF and the standard PF becomes more pronounced as the value of probability increases.

\begin{figure}[!t]
	\centering
	\includegraphics[width=3in,height=2.5in]{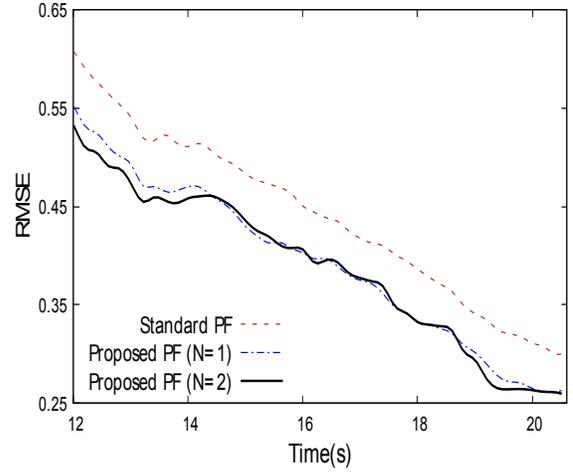}
	\caption{RMSE vs time for state $x_{2,k}$ considering $\hat{p}=0.46$ (Problem 2)}
	\label{fig_9}
\end{figure}
\begin{figure}[!t]
	\centering
	\includegraphics[width=3in,height=2.5in]{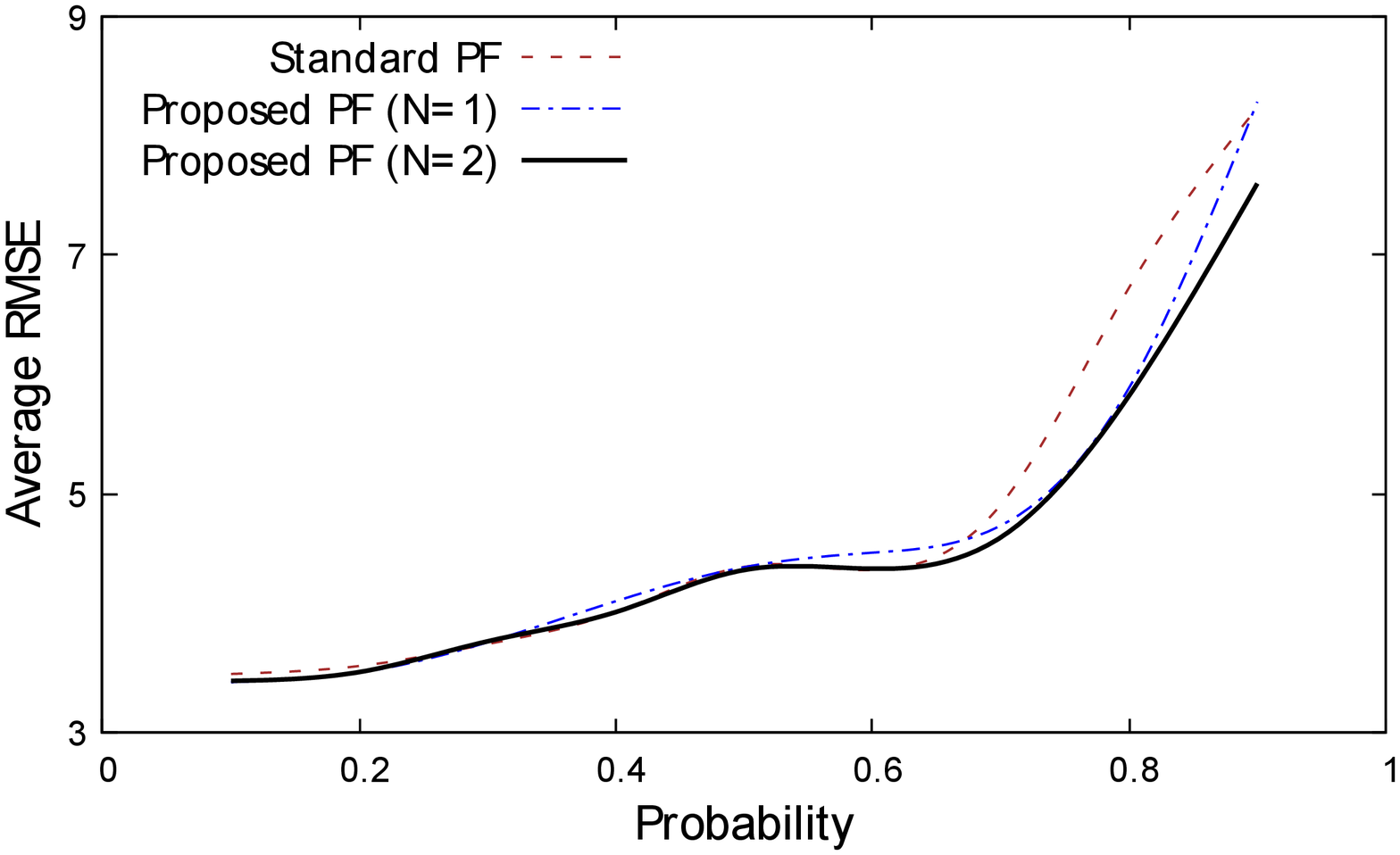}
	\caption{Average RMSE vs probability for state $x_{1,k}$ (Problem 2)}
	\label{fig_10}
\end{figure}
\begin{figure}[!t]
	\centering
	\includegraphics[width=3in,height=2.5in]{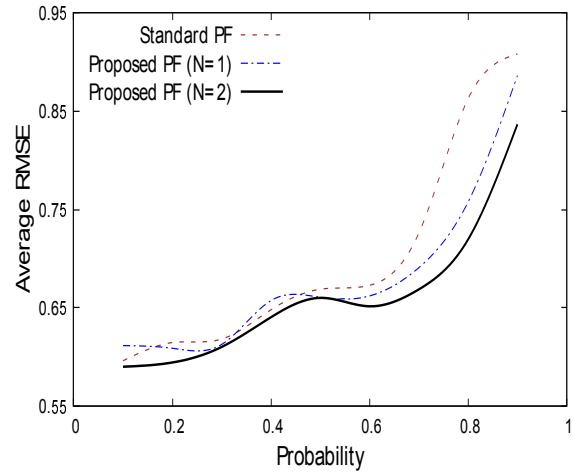}
	\caption{Average RMSE vs probability for state $x_{2,k}$ (Problem 2)}
	\label{fig_11}
\end{figure}

\section{Conclusion and discussion}

 The conventional PF loses its applicability if measurements are randomly delayed at the receiver end. This might be quite frequent in a system where distance between sensor and receiver matters or if its communication bandwidth is limited.

In this technical work, a recursion equation of importance weight has been developed stochastically in accordance with the delay in measurements. A practical measurement model based on i.i.d. Bernoulli random variables has been adopted which includes possibility of random delays in receiving observations along with packet drop situation if any measurement suffers a delay more than a chosen maximum possible delay. Moreover, latency probability of received measurements is usually unknown in practical cases. Hence, this paper presents a method to identify it in randomly delayed measurements environment. Further, it explores the conditions which ensure the convergence of the developed PF and subsequent trade-off it introduces in selecting maximum number of delays. 

To validate the performance of the proposed filter and showcase its superiority, two numerical examples are simulated using the standard PF, the proposed PF with wrong selection of maximum delay and the proposed PF with the correct value of maximum delay. Simulation results show that a filer designed for delayed measurements, even considering less delay than the actual, performs better than the conventional PF. If random delay is more likely in a system, number of maximum possible delay should be chosen such a way that it strikes a balance between avoiding information loss and minimizing convergence error.


%




\ifCLASSOPTIONcaptionsoff
  \newpage
\fi



%
%
%

\medskip
\bibliographystyle{ieeetr}
\bibliography{Particle_filter_paper1}
%
%







\end{document}